\pdfoutput=1

\documentclass[11pt]{article}
\usepackage[a4paper, total={16cm, 24cm}]{geometry}

\usepackage[table]{xcolor}
\renewcommand\tableofcontents{\listoftoc*{toc}} 

\usepackage{authblk}
\author[1]{Piotr Faliszewski}
\author[2]{Martin Lackner}
\author[3]{\authorcr Krzysztof Sornat} 
\author[1]{Stanisław Szufa}
\affil[1]{AGH University, Poland} 
\affil[2]{TU Wien, Austria} 
\affil[3]{IDSIA, USI-SUPSI, Switzerland}

\usepackage[pagebackref]{hyperref}
\hypersetup{
	pdfencoding=auto, 
	psdextra,
	colorlinks=true,
	citecolor=green!50!black,
	linkcolor=red!60!black,
}

\usepackage{natbib}  
\usepackage{caption} 
\usepackage{algorithm}
\usepackage{algorithmic}

\usepackage{graphicx} 
\usepackage{microtype}
\usepackage{amsmath}
\usepackage{amssymb}
\usepackage{subcaption}
\usepackage{nicefrac}
\usepackage[textsize=tiny]{todonotes}
\usepackage{amsthm}
\usepackage{booktabs}
\usepackage{paralist}
\usepackage{thm-restate}

\usepackage[pagebackref]{hyperref}
\hypersetup{
		pdfencoding=auto, 
		psdextra,
		colorlinks=true,
		citecolor=green!40!black,
		linkcolor=red!50!black,
		urlcolor=blue!80!black
	}

\usepackage[nameinlink]{cleveref}
\usepackage{newfloat}
\usepackage{listings}
\DeclareCaptionStyle{ruled}{labelfont=normalfont,labelsep=colon,strut=off} 
\floatstyle{ruled}
\newfloat{listing}{tb}{lst}{}
\floatname{listing}{Listing}

\usepackage{times}
\usepackage{soul}
\usepackage{url}
\usepackage{graphicx}
\usepackage{amsmath}
\usepackage{amsthm}
\usepackage{booktabs}
\usepackage{algorithm}
\usepackage{algorithmic}
\usepackage[switch]{lineno}
\usepackage{collcell}
\usepackage{booktabs}

\usepackage{amsfonts}
\usepackage{mathdots}
\usepackage{subcaption}
\usepackage{amssymb}
\usepackage{nicefrac}

\newtheorem{proposition}{Proposition}

\newtheorem{claim}{Claim}

\newtheorem{example}{Example}

\newtheorem{definition}{Definition}
\newtheorem{lemma}{Lemma}
\newtheorem{theorem}{Theorem}

\newcommand{\disc}{{\mathrm{disc}}}
\newcommand{\jacc}{{\mathrm{jac}}}
\newcommand{\hamm}{{\mathrm{ham}}}
\newcommand{\normhamm}{{\mathrm{nham}}}

\newcommand{\sainte}{Sainte-Lagu\"e}
\newcommand{\seqphrag}{seq-Phragm\'en}

\DeclareMathOperator{\poly}{poly}
\DeclareMathOperator{\argmax}{argmax}

\newcommand\mydots{\hbox to 0.8em{.\hss.\hss.}}
\newcommand{\reals}{{\mathbb{R}}}

\usepackage{todonotes}

\newcommand{\kso}[1]{\textcolor{blue!85}{kso: #1}}

\definecolor{babyblue}{rgb}{0.54, 0.81, 0.94}
\newcommand*{\MinNumber}{0}%
\newcommand*{\MaxNumber}{1}%

\newcommand{\ApplyGradient}[1]{%
  \pgfmathsetmacro{\PercentColor}{100.0*(1-(#1-\MinNumber)/(\MaxNumber-\MinNumber)*(#1-\MinNumber)/(\MaxNumber-\MinNumber))}%
  \edef\x{\noexpand\cellcolor{babyblue!\PercentColor}}\x\textcolor{black}{#1}%
}

\newcolumntype{R}{>{\collectcell\ApplyGradient}{c}<{\endcollectcell}}

\title{An Experimental Comparison of Multiwinner Voting Rules \\ on Approval Elections}

\begin{document}

\maketitle

\begin{abstract}  
  In this paper, we experimentally compare major approval-based
  multiwinner voting rules. To this end, we define a measure of
  similarity between two equal-sized committees subject to a given
  election. Using synthetic elections coming from several
  distributions, we analyze how similar are the committees provided by
  prominent voting rules. Our results can be visualized as ``maps of
  voting rules'', which provide a counterpoint to a purely axiomatic
  classification of voting rules. 
  The strength of our proposed method is its independence from preimposed classifications (such as the satisfaction of concrete axioms),
  and that it indeed offers a much finer distinction than
  the current state of axiomatic analysis.
\end{abstract}

\section{Introduction}
Multiwinner voting is the process of selecting a fixed number of
candidates (a \emph{committee}) based on the preferences of
agents. This general task occurs in a wide range of applications such
as group recommendations \cite{budgetSocialChoice,gawronusing},
blockchain protocols \cite{cevallos2020verifiably}, political
elections \cite{renwick16,bri-las-sko:c:apportionment}, and the
design of Q\&A platforms \cite{israel2021dynamic}. As it is apparent
from this diverse list of applications, there is a multitude of
desiderata for multiwinner voting rules and their desirability
depends on the setting in which they are  applied.
\citet{elk-fal-sko-sli:c:multiwinner-rules} and \citet{FSST-trends}
suggest a classification of multiwinner voting rules based on three
principles: proportionality, diversity and individual excellence.
These principles capture three main goals in multiwinner voting: (i)
to find a committee that reflects the voters' preferences in a
proportional manner, (ii) to find a committee that represents (or, covers) the
opinions of as many voters as possible (diversity), and (iii) to find a
committee that contains the objectively ``best'' candidates
(individual excellence).

In recent years, a large body of work has helped to shape our
understanding of multiwinner voting rules (as surveyed
by \citet{FSST-trends} and \citet{lac:sko:t:approval-book}).
The key method employed
here is axiomatic analysis.  Among the three aforementioned
principles, proportionality has received the most attention and there
now exists a hierarchy of proportionality axioms by which the
proportionality of a voting rule can be assessed
\cite[Chapter~4]{lac:sko:t:approval-book}. Far fewer axioms
exist for diversity and individual excellence
\cite{FaliszewskiSST17,subiza2017representative,fal-tal:balancing-cc,lac-sko:j;multiwinner-approval-axioms}.

While the axiomatic method has been fundamental in advancing our
understanding of multiwinner rules, significant questions cannot be
approached with axioms.  First and foremost, the satisfaction of an
axiom is a binary fact. If an axiom is not satisfied by a voting rule,
this might be due to a fundamental incompatibility that is evidenced
in nearly every election or simply due to an involved counterexample
that hardly occurs in practice (an extreme example is provided in the
work of \citet{BCKLNSST13}, which disproved an over 20-year-old
conjecture by giving a counterexample with $10^{136}$ candidates;
however, a much smaller example was later reported by \citet{BS13}).
Secondly, while the axiomatic approach is useful to highlight
differences between voting rules, it is rarely helpful to establish
similarities.  Even voting rules that do not share (known) axiomatic
properties can behave very similarly on sampled or real-world
preference data.  Identifying similar voting rules is important, e.g.,
if a computationally demanding rule is infeasible in a given setting
and has to be replaced with a faster-to-compute one. 

The goal of our paper is to close this gap in our understanding and
provide a principled method to assess the similarity of voting rules.
Our proposed method does not rely on preimposed classifications (such
as the proportionality/diversity/individual excellence trichotomy) and, instead,
we base our analysis on comparing the committees provided by the rules.
Specifically, given a distance measure between candidates (which we
interpret as a measure of their similarity), we extend it to
committees and compare how close are the committees output by a number
of major multiwinner rules on several families of synthetic
elections.  For the visualization of our results, we adapt the map
framework proposed by \citet{szu-fal-sko-sli-tal:c:map} and
\citet{boe-bre-fal-nie-szu:c:compass} to display voting rules instead
of elections. 

Our distances between candidates depend solely on which voters
approve them. For example, if two candidates are approved by the same
voters, then we view them as being at distance zero, and if the voters
can be partitioned into those approving either one candidate or the
other, then we view these candidates as maximally distant. To measure
the distance between two committees, we first build a bipartite graph
with members of one committee on the left and members of the other
committee on the right, where each two candidates (from the other
committees) are connected by an edge whose weight is equal to their
distance. The distance of the committees is the weight of the
minimum-weight perfect matching in this graph.  Our main findings are
as follows:
\begin{enumerate}
\item We show that computing two most distant committees in a given
  election is intractable in many settings. This is somewhat
  unfortunate, as such committees would be useful to normalize the
  distances between committees. Due to this hardness result, in the
  remainder of the paper we normalize by observed maximum distances.

\item We compute the committees output by various multiwinner rules on
  a number of synthetic elections and compute the average distances
  between committees provided by different rules. We find that
  committees provided by Multiwinner Approval Voting (AV),
  Chamberlin--Courant (CC), and Minimax Approval Voting tend to be the most
  distinct ones, while those provided by proportional rules are
  between those of AV and CC (often, but not always, closer to AV),
  and far away from those of Minimax.

\item For a number of proportionality axioms, we report how often the
  committees output by various rules satisfy them. 
  Surprisingly, our experiments show that the committees provided by
  proportional voting rules (on sampled profiles) generally satisfy stronger proportionality 
  properties than their axiomatic analysis reveals.
  In particular, we find that some axiomatic distinctions are not observable in our data set of 6000 instances.
\end{enumerate}
All in all, we find that all proportional rules---including those that
fail some of the stronger proportionality axioms---are more or less
similar to each other and form a well-defined cluster.



\section{Preliminaries}
For a given graph $G = (V(G), E(G))$ and a vertex $x \in V(G)$, by
$N(x)$ we denote the neighbors of $x$ 
(i.e., the set $\{y \in V(G): (x,y) \in E(G)\}$)
and by $\deg(x)$ we denote the degree of $x$ (i.e., $\deg(x)=|N(x)|$).

\subsection{Elections}
We consider the approval preference model. An election $E = (C,V)$
consists of a set $C = \{c_1, \ldots, c_m\}$ of candidates and a
collection $V = (v_1, \ldots, v_n)$ of voters, where each voter $v_i$
is endowed with a set of candidates that they approve. 
For a
voter $v$ we denote his or her approval set as $A(v)$, and for a
candidate $c$ we write $A(c)$ to denote the set of voters that
approve~$c$. Whenever we use this notation, the election in question
will be clear from the context.

\subsection{Multiwinner Voting Rules}\label{sec:multiwinnerrules}
An approval-based multiwinner voting rule is a function~$f$ that given
an election $E = (C,V)$ and committee size $k$, $k \leq |C|$, provides
a family $f(E,k)$ of size-$k$ subsets of $C$, referred to as the
winning committees.

We give a brief overview of major approval-based multiwinner voting
rule.  We omit technical details and refer to the survey by
\citet{lac:sko:t:approval-book} for details. We assume that all
voting rules use a tiebreaking mechanism to ensure that they return
exactly one winning committee (this property is known as
resoluteness).\footnote{On the technical level, in our experiments we
  use the implementations of the rules provided in the
  \texttt{abcvoting} library~\cite{abcvoting}, together with their
  default tie-breaking.}

\emph{Thiele methods} are an important class of approval-based
multiwinner voting rules that select a committee $W$ maximizing value
$\textit{sc}_w(W)=\sum_{v\in V} \sum_{i=1}^{|A(v)\cap W|} w(i)$ for
some score function~$w$. In particular, \emph{Multiwinner Approval
  Voting (AV)} is defined by $w(i)=i$ (equivalently, AV selects the
$k$ candidates that are approved by the largest number of
voters). This rule is considered to be the prime example of the
principle of \emph{individual excellence}.  \emph{Chamberlin--Courant (CC)}
selects the committee $W$ that maximizes the number of voters that
approve at least one candidate in $W$, i.e., $w(1)=1$ and $w(i) = 0$
for all $i \neq 1$. CC is an example of a rule designed to achieve
\emph{diversity}.  \emph{Proportional Approval Voting (PAV)} and
\emph{\sainte{} Approval Voting (SLAV)} are defined by
$w(i)=\frac{1}{i}$ and $w(i)=\frac{1}{2i-1}$, respectively.
Finally, \emph{$p$-Geometric rules} are defined by $w(i)=p^{-i}$ for
positive constants $p$.
The weights of PAV can be viewed as the most proportional 
choice in this spectrum \cite{justifiedRepresentation,Sanchez-Fernandez2017Proportional};
weight functions that decrease more quickly than that of PAV (such as those used by SLAV and $p$-Geometric rules) give more importance 
to small groups (\textit{degressive} proportionality).

\emph{Sequential Thiele methods}, also referred to as greedy Thiele
rules in the literature, iteratively build a committee by adding the
candidate that increases the score $\textit{sc}_w$ most, starting with
the empty committee and iterating until $k$ candidates are
selected. We consider seq-PAV, seq-SLAV, and seq-CC.  \emph{Greedy
  Monroe}~\cite{sko-fal-sli:j:multiwinner} is a rule similar to
seq-CC, with an additional constraint that each committee member can
represent at most $\nicefrac{n}{k}$ many voters.


\emph{Satisfaction Approval Voting (SAV)} is defined similarly to
AV and
selects a committee~$W$ maximizing
$\sum_{v\in V} \frac{|A(v)\cap W|}{|A(v)|}$.  \emph{Minimax Approval
  Voting (MAV)} \cite{minimaxProcedure} selects a committee~$W$ that
minimizes $\max_{v\in V} |A(v) \setminus W| + |W \setminus A(v)|$,
i.e., it minimizes the maximum Hamming distance between a voter and
the chosen committee.

Finally, we describe the intuition behind two slightly more complex
rules: \emph{Sequential Phragm\'en (\seqphrag)}
\cite{mathprog-phragmen} and the method of \emph{Equal Shares}
\cite{peters2020proportionality}. Both can be understood as mechanisms
where voters use (virtual) budget to jointly pay for the selection of
candidates in the committee. The cost of adding a candidate is set
(arbitrarily) to~$1$. With \seqphrag, voters start with a budget
of~$0$, which is increased continuously until a group of voters can
pay for the first candidate. The cost (of~$1$) is shared among the
members of the group and the process repeats until the committee is
filled. With Equal Shares, voters start with a budget of
$\nicefrac{k}{n}$. Each round, a candidate is selected that requires
the least budget per voter. This procedure may result in fewer than
$k$ committee members (since the budget is fixed in advance), in which
case \seqphrag{} is used to fill the committee. We refer to the works
of \citet{abcvoting} and \citet{peters2020proportionality} for
details.

\subsection{Proportionality}
Much of the recent progress in multiwinner voting has been dedicated
to the concept of proportionality, to formally define proportionality
and to identify voting rules behaving proportionally. Since our
framework enables us to identify similar voting rules, we are able to
ask whether ``proportional'' rules are indeed similar.  To this end,
we consider four proportionality axioms. These axioms apply to
committees and a voting rule is said to satisfy such a property if it
is guaranteed to return committees exhibiting this property.

\begin{definition}[\citealt{Sanchez-Fernandez2017Proportional}]
  Given an election $E=(C,V)$, a committee~$W$ of size~$k$ satisfies
  \emph{Proportional Justified Representation (PJR)} if there is no
  group of agents $N\subseteq V$ of size
  $|N|\geq \ell\cdot\frac{n}{k}$ that jointly approves at least $\ell$
  common candidates and $\left| \bigcup_{v\in N} A(v) \cap W \right|< \ell$.
\end{definition}
PJR is satisfied by PAV,
\seqphrag, Equal Shares, and Greedy Monroe\footnote{It satisfies PJR
  if $k$ divides $n$ \cite{Sanchez-Fernandez2017Proportional}.}
(among the rules introduced in Section~\ref{sec:multiwinnerrules}).

\begin{definition}[\citealt{justifiedRepresentation}]
  Given an election $E=(C,V)$, a committee~$W$ of size~$k$ satisfies
  \emph{Extended Justified Representation (EJR)} if there is no group
  of agents $N\subseteq V$ of size $|N|\geq \ell\cdot\frac{n}{k}$ that
  jointly approves at least $\ell$ common candidates and
  for each voter $v \in N$, 
  $\left|  A(v) \cap W \right|< \ell$.
\end{definition}
EJR strengthens PJR and is satisfied only by PAV and Equal Shares.
Finally, Justified Representation (JR) is the special case of PJR (and EJR)
restricted to $\ell=1$.

Priceability \cite{peters2020proportionality} is a notion of
proportionality based on assigning budget to voters.
\begin{definition}
  Given an election $E=(C,V)$, a committee~$W$ satisfies
  \emph{priceability} if there is a budget $b\geq 0$ and for each
  voter~$v_i$ a spending function $b_i \colon C \to \mathbb{R}_{\geq 0}$ such
  that:
\begin{enumerate}
\item for each $v_i \in V$, $\sum_{c \in C}b_i(c) \leq b$ (voters do
  not spend more than $b$),
\item for each $v_i \in V$, if $c \notin A(v_i)$ then $b_i(c) = 0$ (voters do not pay for
  candidates they do not approve),
\item if $c \in W$ then $\sum_{v_i \in V}b_i(c) = 1$ and otherwise this
  value is $0$ (payments only for committee members), and
\item for each $c \notin W$,
  $\sum_{v_i \in A(c)}\left(b - \sum_{c' \in W}b_i(c')\right) \leq 1$
  (there is no candidate left that could be bought with \emph{more}
  than a budget of $1$).
\end{enumerate}
\end{definition}
Priceability establishes fairness by distributing the same amount of
budget to each voter and requiring that no candidate could be bought
with a smaller amount (condition~4).  Priceability implies PJR and is
incomparable to EJR. It is satisfied by \seqphrag{} and Equal Shares.



\subsection{Statistical Cultures}

Below, we define the statistical cultures that we use for generating
synthetic instances of approval elections.
\begin{description}
\item[Resampling Model.] The resampling model is parameterized by two values,
  $p \in [0,1]$ and $\phi \in [0,1]$, where the first one describes
  the average number of approvals in a vote, and the second one
  describes the level of disturbance. To generate an election, we
  first sample a central ballot by approving uniformly at random
  $\lfloor pm \rfloor$ candidates. Then, to generate a vote, we
  proceed as follows. For each candidate in the central ballot, we add
  him or her to the vote with probability $1-\phi$, and with
  probability $\phi$ we resample that candidate, that is, we add him
  or her to the ballot with probability $p$. For each candidate
  outside the central ballot, with probability $1-\phi$, we keep him
  or her outside the ballot, and with probability $\phi$ we resample
  him or her.

\item[Disjoint Model.] The disjoint model has three parameters,
  $p \in [0,1]$, $\phi \in [0,1]$, and
  $g\in \mathbb{N} \setminus \{0\}$. It is similar to the resampling
  model, but instead of a single central ballot we have $g$ central
  ballots. At the beginning, we sample $g$ disjoint central ballots
  uniformly at random (each of them approving $\lfloor pm \rfloor$
  candidates, and then proceed in the same manner as before, however
  before generating each vote, we uniformly at random select one of
  our $g$ central ballots on which our vote will be based.

\item[Euclidean Models.] The Euclidean model is parameterized by a radius
  $r\in \mathbb{R}_+$. For each voter and each candidate we sample their
  ideal point in the $t$-dimensional Euclidean space (uniform in a
  $t$-dimensional cube). Then, each
  voter approves all the candidates within radius~$r$.

\item[Party-list Model.] The party-list model is parameterized by $\alpha$
  and $g$. First, we divide the candidates into $g$ groups of size
  $\lfloor \nicefrac{m}{g} \rfloor$ each; we refer to these groups as
  the parties. (If $m \ \textrm{mod} \ g \neq 0$ then some
    candidates remain unapproved.) Second, we create an urn which
  contains a single vote for each party, approving exactly its
  members. Then, we iteratively generate votes (one at a time): 
  We draw a vote from the urn, add its copy to the election,
  and return the vote to the urn together with $\alpha g$ copies.


\item[Pabulib Model.] We also use real-life participatory budgeting (PB) data from Pabulib~\cite{fal-fli-pet-pie-sko-sto-szu-tal:c:participatory-budgeting-tools}, which we treat as a statistical culture over approval elections (in particular, we omit details related to PB, such as the costs of the candidates). We selected $21$ instances, i.e., all that contain at least $100$ candidates and $100$ voters, and where the average number of approved candidates per voter is at least $3$ (for instances with truncated ordinal ballots, we convert them to approval ones by approving all candidates that were ranked). To sample an election from Pabulib, we first uniformly at random choose one of the original $21$ instances, then we randomly select a subset of $100$ candidates and a random subset of $100$ voters. We omit voters who cast empty ballots.

\end{description}

The resampling and disjoint models are due to
\citet{SFJLSS22}. Various Euclidean models are commonly used in the
literature on elections~\cite{enelow1984spatial,enelow1990advances};
we point to the recent works of \citet{bre-fal-kac-nie2019:experimental_ejr} and \citet{god-bat-sko-fal:c:2d}, which use
them in the approval setting. The party-list model is an adaptation of
the Polya-Eggenberger urn model~\cite{berg1985paradox,mcc-sli:j:similarity-rules}.

\section{Similarity Between Committees}
Our idea of comparing voting rules is based on measuring the
similarity between the committees that they produce. 

\subsection{Basic Framework}
Fix some election $E = (C,V)$ with candidate set
$C = \{c_1, \ldots, c_m\}$ and voter collection
$V = (v_1, \ldots, v_n)$. We assume that we have some distance $d$
over the candidates, such that if $d(c_i,c_j)$ is small---for whatever
``small'' means under a given distance---then candidates $c_i$ and
$c_j$ are similar, and if it is large then they are not.  In the most
basic setting we could take the discrete distance, where
$\disc(c_i,c_j) = 0$ exactly if $i = j$ and $\disc(c_i,c_j) = 1$
for every other case. Later we will discuss two more 
distances.

Let $X = \{x_1, \ldots, x_k\}$ and $Y = \{y_1, \ldots, y_k\}$ be two
size-$k$ committees over $C$. We extend $d$ to act on committees as
follows: We form a bipartite graph with members of $X$ as the vertices
on the left, members of $Y$ as the vertices on the right, and where
for each $x \in X$ and each $y \in Y$ we have an edge with weight
$d(x,y)$; if some candidate $c$ belongs to both $X$ and $Y$, then we
have two copies of $c$, one on the left and one on the right.  The
distance $d(X,Y)$ between $X$ and $Y$
is the weight of the minimum-weight matching in this graph.  One can
verify that it indeed is a pseudodistance.

\newcommand{\widehhat}[1]{#1}
\begin{proposition}
  For each (pseudo)distance $d$ over the candidates, its
  above-described extension to committees is a pseudodistance.
\end{proposition}
\begin{proof}
  Let $X = \{x_1, \ldots, x_k\}$, $Y = \{y_1,\ldots, y_k\}$, and
  $Z = \{z_1,\ldots, z_k\}$ be three size-$k$ committees from the same
  election. By definition of $\widehhat{d}$, we immediately get that
  $\widehhat{d}(X,X) = 0$ and $\widehhat{d}(X,Y) = \widehhat{d}(Y,X)$. It
  remains to show that
  $\widehhat{d}(X,Y) \leq \widehhat{d}(X,Z) + \widehhat{d}(Z,X)$. By reordering the
  members of $X$, $Y$, and $Z$, we can assume that:
  \begin{align*}
    \widehhat{d}(X,Z) &= d(x_1,z_1) + \cdots + d(x_k,z_k), \text{ and} \\
    \widehhat{d}(Z,Y) &= d(z_1,y_1) + \cdots + d(z_k,y_k).    
  \end{align*}
  Since $d$ is a distance, for each $i \in [k]$ we have
  $d(x_i,y_i) \leq d(x_i,z_i) + d(z_i,y_i)$. Further, we know that
  $\widehhat{d}(X,Y) \leq d(x_1,y_i) + \cdots + d(x_k,y_k)$ (because
  instead of using the lowest-weight matching we use some fixed
  one). By putting these inequalities together, we get that, indeed,
  $\widehhat{d}(X,Y) \leq \widehhat{d}(X,Z) + \widehhat{d}(Z,Y)$.
\end{proof}

Let us now consider concrete distance measures between candidates, 
all of which are well-known distances (cf. \cite{deza2009encyclopedia}). 
The discrete
distance is a rather radical approach as it stipulates that no two
candidates are ever similar to each other. Yet, one could argue that
there are cases where some candidates are clearly more similar to each
other than others.
\begin{example}\label{ex:similarity-intuition}
  Consider an election $E = (C,V)$ with $C = \{p,q,r,s\}$ and voter
  collection $V = \{v_1, v_2, v_3, v_4)$. The approval sets are as
  follows:
  \begin{align*}
    A(v_1) &= \{p\}, &
    A(v_2) &= \{q,s\}, \\
    A(v_3) &= \{p,r,s\}, &
    A(v_4) &= \{q\}. 
  \end{align*}
  One could argue that candidates $p$ and $q$ are completely
  dissimilar because they are approved by complementary sets of
  voters. On the other hand, half of the voters who approve $p$ also
  approve $r$, and every voter who approves $r$ also approves~$p$. The
  similarity between $p$ and $s$ seems a bit weaker than that between
  $p$ and $r$ because there is a voter who approves $s$ but does not
  approve $p$.
\end{example}
In addition to the discrete one, we consider the following two
distances between the candidates (we assume some election $E$ with $n$
voters):
\begin{enumerate}
\item The Hamming distance between candidates $c$ and $d$, denoted by $\hamm(c,d)$, is the
  number of voters that approve one of them but not the
  other. Formally, we have:
  \[
  \hamm(c,d) = |A(c) \setminus A(d)| + |A(d) \setminus A(c)|.
  \]
\item The Jaccard distance between candidates $c$ and $d$ is defined
  as: 
  \[
  \jacc(c,d) = 1 - \frac{|A(c) \cap A(d)|}{|A(c) \cup A(d)|}=\frac{\hamm(c,d)}{|A(c) \cup A(d)|}
  \]
\end{enumerate}
There are two main differences between the Hamming and the Jaccard
distances.  The less important one is that the former is not
normalized and can assume values between $0$ and~$n$, whereas the
latter always assumes values between $0$ and~$1$.  Occasionally, we
will speak of normalized Hamming distance
$\normhamm(c,d) = \nicefrac{1}{n}\cdot\hamm(c,d)$.

The more important difference is in how both distances interpret lack
of an approval for a candidate. Under the Hamming distance, we assume
that it is a conscious decision, indicating that a voter disapproves
of a candidate. Hence, if neither candidate $c$ nor $d$ is approved by
a voter, then Hamming distance interprets it as a sign of similarity
between them.\footnote{Yet, a voter might dislike $c$ and $d$ for two
  different reasons---e.g., one is too liberal and the other is not
  liberal enough, which makes this approach questionable.}  Under the
Jaccard distance, we assume that we cannot make any such conclusions.
\begin{example}\label{ex:similarity}
  Let us consider the election from
  Example~\ref{ex:similarity-intuition}.  We have the following
  normalized Hamming and Jaccard distances between the candidates:
  \begin{align*}
    \normhamm(p,q)&= 1, &
    \jacc(p,q) &= 1,  \\
    \normhamm(p,r)&= \nicefrac{1}{4}, &
    \jacc(p,r) &= \nicefrac{1}{2},  \\
    \normhamm(p,s)&= \nicefrac{1}{2}, &
    \jacc(p,s) &= \nicefrac{2}{3},  \\
    \normhamm(q,r)&= \nicefrac{3}{4}, &
    \jacc(q,r) &= 1.
  \end{align*}
  As our intuition from Example~\ref{ex:similarity-intuition}
  suggested, both our distances indicate that candidate $p$ is
  completely dissimilar from candidate $q$, and is more similar to $r$
  than to $s$. However, according to the normalized Hamming distance
  $q$ and $r$ are not completely dissimilar, whereas they are at
  maximum distance according to Jaccard.
\end{example}

\subsection{Hardness of Finding Farthest Committees}
Given an election and a committee size, it is convenient to know what
is the largest possible distance between two committees. In
particular, one could use this value to normalize distances between
committees provided by various rules. Unfortunately, it turns out that
finding two farthest committees under a given distance is often
intractable and in our experiments we will resort to
normalizing by largest observed distances.
We define the {\scshape Farthest Committees} problem as follows.

\begin{definition}
In the {\scshape Farthest Committees} problem (FC)
we are given an election $(C,V)$, an integer $k$, and a
pseudodistance $d$ over the candidates.
The goal is to output two committees,
each of size $k$, which maximize the distance $d$ between them.
\end{definition}
FC can be considered also in its decision variant, where we ask if there
are two size-$k$ committees whose distance is at least a given value.

FC under the discrete distance is trivially polynomial-time solvable.
In Theorem~\ref{thm:fc-np-hard} we show that FC is computationally
hard in the cases of the Hamming distance and the Jaccard distance.
This comes through a reduction from the {\scshape Balanced Biclique}
problem defined as follows.

\begin{definition}
In the {\scshape Balanced Biclique} problem (BB) we are given
a bipartite graph $G=(A_G \cup B_G,E_G)$ and an
integer $k_G$.  The question is whether there exists a biclique of
size $k_G \times k_G$ in $G$ ($k_G$-biclique), i.e., a set of vertices
$X \cup Y$ such that $X \subseteq A_G, Y \subseteq B_G, |X|=|Y|=k_G$
and for every $x \in X, y \in Y$ we have $(x,y) \in E_G$.
\end{definition}
\noindent BB is known to be NP-hard~\cite[p. 196]{GareyJ79} and W[1]-hard
w.r.t. $k_G$~\cite{Lin18}.

The idea of the reduction is to encode vertices by candidates and
edges by large distances (no-edges are represented by slightly smaller
distances).  A technical issue is to define proper votes in order to:
1) achieve required relation between the distances and, 2) forbid taking
candidates to one committee that come from both parts of the input
bipartite graph.

\begin{theorem}\label{thm:fc-np-hard}
  {\scshape Farthest Committees} is NP-hard and W[1]-hard w.r.t. $k$ (the size of committees), even in the case of the Hamming distance or the Jaccard distance.
\end{theorem}
\begin{proof}
  First we show the theorem for the Jaccard distance.  The reduction
  is from {\scshape Balanced Biclique} (BB).  We take an instance
  $(G=(A_G \cup B_G,E_G),k_G)$ of BB and produce an instance of the
  decision version of {\scshape Farthest Committees} (FC) with the Jaccard
  distance.

  For every vertex $x \in A_G \cup B_G$ we define a candidate
  $x \in C$. Thus, $|C| = |A_G \cup B_G|$.  We overload the notation
  and by $x$ we mean both a vertex $x$ from the graph $G$ and a candidate $x$
  from $C$.  This is due to the bijection between them.

  We define the votes as follows.  There is one voter $v_A$ with an
  approval set $A_G$ and one voter $v_B$ with an approval set $B_G$.
  Then, for every $x \in A_G$ we define a voter $v_x$ with an approval
  set $\{x\} \cup (B_G \setminus N(x))$ (recall that $N(x)$ means the
  neighbors of $x$).  Hence, we have $|V| = |A_G|+2$.

  We ask for an existence of two committees of size $k=k_G$ with the
  Jaccard distance equal to $k$ (indeed, $k$ is the largest possible
  Jaccard distance between two size-$k$ committees).

  In order to show correctness of the reduction, it suffices to prove
  the following lemma which implies that only $k$-bicliques correspond
  to committees at Jaccard distance~$k$.


\begin{lemma}\label{lemma:bb-jacc}
  For every $X, Y \subseteq A_G \cup B_G, |X|=|Y|=k$, the following equivalence holds:
  $(X,Y)$ is a $k$-biclique in $G$
  if and only if
  $\jacc(X,Y) = k$.
\end{lemma}

\begin{proof}
We will show two implications. First, we assume that
$X,Y \subseteq A_G \cup B_G$ form a $k$-biclique in $G$.  It means
that, w.l.o.g., $X \subseteq A_G$ and $Y \subseteq B_G$.  We will
show that the two committees represented by candidates from $X$ and
$Y$ realize the Jaccard distance $k$.

First, we show that for any two candidates from different committees there is no voter approving both.
\begin{claim}\label{claim:xy-distinct-approvals}
  For every $x \in \hspace{-1pt} X, y \in Y$ we have $A(x) \cap A(y) = \emptyset$.
\end{claim}
\begin{proof}
  It is enough to show that every voter does not approve one of the
  candidates: $x$ or $y$.  Recall that $x \in A_G$ and $y \in B_G$.
  First, voter $v_B$ does not approve $x$ and voter $v_A$ does not
  approve $y$.  Second, as $(X,Y)$ is a biclique we have $y \in N(x)$,
  therefore voter $v_x$ does not approve $y$.  For every
  $z \in A_G \setminus \{ x \}$, $v_z$ does not approve $x$.
\end{proof}

Therefore, the Jaccard distance between every two candidates $x \in X$
and $y \in Y$ is equal to
$1-\frac{|A(x) \cap A(y)|}{|A(x) \cup A(y)|}
\stackrel{\text{Claim}\ref{claim:xy-distinct-approvals}}{=} 1$.
Hence, any perfect matching between $X$ and $Y$
(so also the Jaccard distance) has weight $|X| = k$, as required.

In order to show the implication in the other direction,
let $X,Y \subseteq C$ be two $k$-sized committees at the
Jaccard distance equal to $k$.  We will show that $(X,Y)$ is a
$k$-biclique.
We will show that $X$ and $Y$ are subsets of different parts of $G$.
For this we use the following claim.

\begin{claim}\label{claim:xy-jacc-1}
  For every $x \in X, y \in Y$ we have $\jacc(x,y) = 1$.
\end{claim}
\begin{proof}
  By contradiction, assume that for some $x \in X$ and $y \in Y$ we
  have $\jacc(x,y) < 1$.  Then, a perfect matching between candidates
  in $X$ and candidates in $Y$ that contains an edge $(x,y)$ (which
  has weight strictly smaller than $1$) and any other edges (which
  have weights at most $1$) has weight strictly smaller than
  $|X| = k$.  This is in contradiction with $\jacc(X,Y) = k$.
\end{proof}

Notice that for every $x_1,x_2 \in A_G$ we have
$|A(x_1) \cap A(x_2)| = 1$ (because of voter $v_A$), hence
$\jacc(x_1,x_2) < 1$.  Analogously, we have
$|A(y_1) \cap A(y_2)| \geq 1$ for every $y_1,y_2 \in B_G$ (because of
voter $v_B$), hence $\jacc(y_1,y_2) < 1$.  It means that, due to
Claim~\ref{claim:xy-jacc-1}, one of the sets $X$ and $Y$ is a subset
of $A_G$ and the other one is a subset of $B_G$.

It remains to show that $(X,Y)$ is a biclique, i.e., for every
$x \in X$ and $y \in Y$ we have $(x,y) \in E_G$.  Let us assume, by
contradiction, that $(x,y) \notin E_G$, i.e., $y \notin N(x)$.  It
means that voter $v_x$ approves both candidates, $x$ and $y$.
Therefore, $|A(x) \cap A(y)| \geq 1$ (actually, this is exactly $1$),
hence $\jacc(x,y) < 1$.  This is in contradiction with
Claim~\ref{claim:xy-jacc-1}, so $(X,Y)$ is truly a $k$-biclique.
\end{proof}


It is straightforward to see that the reduction runs in polynomial time.
Hence, NP-hardness of FC follows from NP-hardness of
BB~\cite[p. 196]{GareyJ79}. Furthermore, BB is W[1]-hard
w.r.t. $k_G$~\cite{Lin18} and we have $k=k_G$, so indeed, FC is
W[1]-hard w.r.t. $k$.

In the case of the Hamming distance we use the same reduction with the
following changes: (1) Instead of one voter $v_A$ we add $|A_G|$ many
voters $v_A$; (2) Instead of one voter $v_B$ we add $|A_G|$ many voters
$v_B$; (3) We add $|A_G|$ many {\textit degree-correcting} voters who
do not approve any candidates from $A_G$ but, for every candidate
$y \in B_G$, exactly $\deg(y)$ many of them approve $y$; (4) We ask
for an existence of two committees of size $k$ with the Hamming
distance equal to $k \cdot (3|A_G|+1)$.  After such modifications we
still have $C = A_G \cup B_G$, but we have $|V| = 4 \cdot |A_G|$.

In order to show correctness of the reduction it is enough to prove the following lemma.

\begin{lemma}
  For every $X, Y \subseteq A_G \cup B_G, |X|=|Y|=k$, the following equivalence holds:
  $(X,Y)$ is a biclique of size $k \times k$ in $G$
  if and only if
  $\hamm(X,Y) = k \cdot (3|A_G|+1)$.
\end{lemma}
\begin{proof}
  The proof is analogous to that of Lemma~\ref{lemma:bb-jacc}.  The
  main difference is in calculating pairwise distances between
  candidates because we have a different set of voters and we use the
  Hamming distance instead of the Jaccard one.
  
  The distances between candidates are as follows  (in the
  calculations below, the first term contains the Hamming
  distance realized by all the voters $v_A$ and $v_B$, the second one regards
  all the voters $v_x, x \in A_G$, and the third one regards the
  degree-correcting voters):
  \begin{enumerate}
    \item If $x_1,x_2 \in A_G$ then $\hamm(x_1,x_2) = (0)+(2)+(0) = 2$.
    \item If $y_1,y_2 \in B_G$ then $\hamm(y_1,y_2) \leq
     (0) + (|A_G|) + (|A_G|) = 2 \cdot |A_G|$.
    \item If $x \in A_G, y \in B_G, (x,y) \notin E$ then $\hamm(x,y) = (2|A_G|)+(|A_G|-\deg(y)-1)+(\deg(y)) = 3 \cdot |A_G|-1$.
    \item If $x \in A_G, y \in B_G, (x,y) \in E$ then $\hamm(x,y) = (2|A_G|)+(|A_G|-\deg(y)+1)+(\deg(y)) = 3 \cdot |A_G|+1$.
  \end{enumerate}
  Using the distances calculated above and analogous arguments as in
  the proof of Lemma~\ref{lemma:bb-jacc}, we obtain that two
  committees corresponding to two parts of a $k$-biclique realize the
  required distance $k \cdot (3|A_G|+1)$.  Furthermore, the required
  distance between committees can be realized only by two committees
  where all pairs of candidates (one candidate from one committee) are
  at the Hamming distance $3 \cdot |A_G|+1$.  It means that each such
  pair defines an edge in the original graph, hence a $k$-biclique.
\end{proof}

This finishes the proof of the theorem.
\end{proof}

FC with the normalized Hamming distance is also NP-hard because its
objective function is just scaled compared to the Hamming distance.

\subsubsection{Parameterized Hardness}
A brute-force algorithm for {\scshape Farthest Committees} (for any distance measure) runs in time $m^{2k} \cdot \poly(n,m)$ (by evaluating all possible pairs of $k$-sized committees).
Unfortunately, we cannot hope for large improvements over this as stated below.

\begin{proposition}\label{prop:eth-hard-nmk}
  Under a randomized version of the Exponential Time Hypothesis, there is no $(n+m)^{o(\sqrt{k})}$-time algorithm for {\scshape Farthest Committees}.
  This holds even for the Hamming distance and for the Jaccard distance.
\end{proposition}
\begin{proof}
The result follows from hardness of {\scshape Balanced Biclique}~\cite{Lin18},
which excludes (under the same hypothesis) existence of
$|V_G|^{o(\sqrt{k})}$-time algorithm.
Using our reduction from Theorem~\ref{thm:fc-np-hard} and the fact that $n = O(|V_G|), m = O(|V_G|)$, the theorem follows.
\end{proof}
%
Naturally, FC is FPT w.r.t. $m$ (for any distance measure) by a
brute-force algorithm which evaluates $\binom{m}{k}^2 \leq 4^m$ many
pairs of committees.

Next, we show that FC is FPT w.r.t. $n+k$,
although the dependence on a parameter is double-exponential.
\begin{theorem}\label{thm:fpt-n-k}
  For a given candidate pseudodistance function $d(x,y)$,
  which is a function of $A(x)$ and $A(y)$ only,
  {\scshape Farthest Committees} under $d$ is FPT w.r.t. $n+k$.
\end{theorem}
\begin{proof}
We say that two candidates $x,y$ have the same type if they are approved by the same voters, i.e., $A(x) = A(y)$.
Let us denote the number of candidate types by $t$.
Clearly, we have $t \leq 2^n$.
We find a solution by guessing how many candidates of every type are in both committees.
For one committee there are at most $k^t \leq k^{2^n}$ many choices.
Hence, we can find a solution after $k^{2^{n+1}}$ many checks, each in polynomial time (by computing the minimum-weight matching).
\end{proof}

Using more involved arguments we show that FC is FPT w.r.t. $n$.
Such algorithm can be useful in the case when the committee size $k$ is much larger than $n$.
\begin{theorem}\label{thm:fpt-n}
  For a given candidate pseudodistance function $d(x,y)$,
  which is a function of $A(x)$ and $A(y)$ only,
  {\scshape Farthest Committees} under $d$ is FPT w.r.t. $n$.
\end{theorem}
The complete proof is in Appendix~\ref{app:fpt-n}.
%
The main technique used to provide this result
is formulating FC as an Integer Linear Program (ILP) with the number
of integer variables and the number of constraints being a function of
$n$.  Then, it is enough to apply the result of
\citet{lenstra1983integer} for solving ILPs.  The main idea for the ILP is
to consider types of candidates---defined for every candidate $c$ as a
subset of voters that approve $c$, i.e., $A(c)$.  There are at most
$2^n$ types of candidates.  Notice that two candidates of the same
type have exactly the same distances to any other candidate.  Hence,
for each candidate type, by two integer variables we encode how many
candidates of this type are included in both committees.  Then, the
objective function is to maximize the weight of a b-matching (a
matching in which a vertex can be matched multiple times) according to
the capacities defined by the integer variables.  The main technical
issue we faced is ensuring that the achieved matching of maximum
weight is also the minimum-weight matching among the chosen
candidates.  We overcome this obstacle by adding $2^{5^n}$ constraints
(one for every cycle in a bipartite graph in which vertices represent
types of candidates).

\newcommand{\ax}{\multicolumn{1}{c}{\textbf{1.0}}}
\newcommand{\axl}{\multicolumn{1}{c|}{\textbf{1.0}}}
\begin{table*}[th]
\scalebox{0.74}{
\footnotesize
\centering
{
\setlength\tabcolsep{3.3pt}
\begin{tabular}{c|RRRR|RRRR|RRRR|RRRR|RRRR|RRRR}
\toprule
Rule & \multicolumn{4}{c|}{Resampling} & \multicolumn{4}{c|}{Disjoint} & \multicolumn{4}{c|}{Party-list} & \multicolumn{4}{c|}{1D} & \multicolumn{4}{c|}{2D} & \multicolumn{4}{c}{Pabulib} \\ \midrule
Equal Shares & \ax & \ax & \ax & \axl & \ax & \ax & \ax & \axl & \ax & \ax & \ax & \axl & \ax & \ax & \ax & \axl & \ax & \ax & \ax & \axl & \ax & \ax & \ax & \ax \\ \midrule
\seqphrag & \ax & 1.0 & \ax & \axl & \ax & 1.0 & \ax & \axl & \ax & 1.0 & \ax & \axl & \ax & 1.0 & \ax & \axl & \ax & 1.0 & \ax & \axl & \ax & 1.0 & \ax & \ax \\ \midrule
PAV & 1.0 & \ax & \ax & \axl & 1.0 & \ax & \ax & \axl & 1.0 & \ax & \ax & \axl & 1.0 & \ax & \ax & \axl & 1.0 & \ax & \ax & \axl & 1.0 & \ax & \ax & \ax \\ \midrule
seq-PAV & 1.0 & 1.0 & 1.0 & 1.0 & 1.0 & 1.0 & 1.0 & 1.0 & 1.0 & 1.0 & 1.0 & 1.0 & 1.0 & 1.0 & 1.0 & 1.0 & 1.0 & 1.0 & 1.0 & 1.0 & 1.0 & 1.0 & 1.0 & 1.0 \\ \midrule
SLAV & 1.0 & 1.0 & 1.0 & \axl & 0.95 & 1.0 & 1.0 & \axl & 0.36 & 1.0 & 1.0 & \axl & 0.88 & 1.0 & 1.0 & \axl & 0.99 & 1.0 & 1.0 & \axl & 1.0 & 1.0 & 1.0 & \ax \\ \midrule
seq-SLAV & 1.0 & 1.0 & 1.0 & 1.0 & 0.95 & 1.0 & 1.0 & 1.0 & 0.36 & 1.0 & 1.0 & 1.0 & 0.86 & 1.0 & 1.0 & 1.0 & 0.99 & 1.0 & 1.0 & 1.0 & 1.0 & 1.0 & 1.0 & 1.0 \\ \midrule
Greedy Monroe & 1.0 & 1.0 & \ax & \axl & 0.88 & 1.0 & \ax & \axl & 0.29 & 1.0 & \ax & \axl &  0.78 & 1.0 & \ax & \axl & 0.93 & 1.0 & \ax & \axl & 0.95 & 1.0 & \ax & \ax \\ \midrule
Geom.~2 & 1.0 & 1.0 & 1.0 & \axl & 0.95 & 1.0 & 1.0 & \axl & 0.14 & 0.55 & 0.55 & \axl & 1.0 & 1.0 & 1.0 & \axl & 1.0 & 1.0 & 1.0 & \axl & 1.0 & 1.0 & 1.0 & \ax \\ \midrule
Geom.~3 & 0.99 & 1.0 & 1.0 & \axl & 0.87 & 0.99 & 0.99 & \axl & 0.03 & 0.3 & 0.3 & \axl & 0.88 & 1.0 & 1.0 & \axl & 0.99 & 1.0 & 1.0 & \axl & 0.97 & 1.0 & 1.0 & \ax \\ \midrule
Geom.~4 & 0.96 & 1.0 & 1.0 & \axl & 0.82 & 0.98 & 0.98 & \axl & 0.02 & 0.23 & 0.23 & \axl & 0.77 & 1.0 & 1.0 & \axl & 0.96 & 1.0 & 1.0 & \axl & 0.93 & 1.0 & 1.0 & \ax \\ \midrule
Geom.~5 & 0.93 & 1.0 & 1.0 & \axl & 0.79 & 0.98 & 0.98 & \axl & 0.01 & 0.2 & 0.2 & \axl & 0.68 & 1.0 & 1.0 & \axl & 0.93 & 1.0 & 1.0 & \axl & 0.88 & 1.0 & 1.0 & \ax \\ \midrule
SAV & 1.0 & 1.0 & 1.0 & 1.0 & 0.5 & 0.64 & 0.64 & 0.64 & 0.06 & 0.08 & 0.08 & 0.08 & 0.0 & 0.01 & 0.01 & 0.01 & 0.0 & 0.06 & 0.06 & 0.06 & 0.84 & 1.0 & 1.0 & 1.0 \\ \midrule
AV & 1.0 & 1.0 & 1.0 & 1.0 & 0.49 & 0.63 & 0.63 & 0.63 & 0.06 & 0.08 & 0.08 & 0.08 & 0.0 & 0.01 & 0.01 & 0.01 & 0.01 & 0.12 & 0.12 & 0.12 & 0.91 & 0.98 & 0.98 & 0.98 \\ \midrule
Minimax AV & 0.71 & 0.99 & 0.99 & 0.99 & 0.46 & 0.89 & 0.89 & 0.98 & 0.0 & 0.05 & 0.05 & 1.0 & 0.0 & 0.09 & 0.09 & 0.09 & 0.02 & 0.28 & 0.28 & 0.28 & 0.59 & 0.86 & 0.86 & 0.86 \\ \midrule
seq-CC & 0.97 & 1.0 & 1.0 & \axl & 0.62 & 0.88 & 0.88 & \axl & 0.0 & 0.06 & 0.06 & \axl & 0.23 & 0.99 & 0.99 & \axl & 0.47 & 0.99 & 0.99 & \axl & 0.58 & 1.0 & 1.0 & \ax \\ \midrule
CC & 0.21 & 0.61 & 0.65 & \axl  & 0.62 & 0.85 & 0.86 & \axl & 0.0 & 0.04 & 0.04 & \axl & 0.23 & 0.98 & 0.99 & \axl & 0.42 & 0.98 & 0.99 & \axl & 0.56 & 1.0 & 1.0 & \ax \\ \midrule
\end{tabular}}
}
\caption{Fraction of instances satisfying priceability, EJR, PJR, and JR---presented respectively (from left to right). If $1.0$ is printed in bold, the corresponding axiom is guaranteed by the voting rule.}
\label{tab:priceability}
\end{table*}

\section{Map of Rules}\label{sec:map-rulez}
In this section, we present our \textit{Map of Multiwinner Voting
  Rules}. The map is constructed in the following way. First, we
generate a number of elections from several statistical cultures
(details are described later in the text). Second, we compute winning
committees under the considered voting rules. We impose resoluteness,
i.e., each rule outputs exactly one winning committee. Third, we compute the
distances between winning committees using the Jaccard distance (we
chose it over Hamming due to the way it interprets the lack of
an approval for a candidate). 
For each election, we normalize the Jaccard distance by dividing it by the largest distance between two committees, outputted by our voting rules, for that election.
For each two rules and each statistical
culture, we compute the average distance between the rules'
committees; this gives us distance matrices for the cultures. We embed these
matrices in 2D-Euclidean space using a variant of the 
algorithm of \citet{kam-kaw:j:embedding}, as implemented by \citet{mt:sapala}. If two rules output similar winning committees, then we
view these two rules as similar; the algorithm places such rules close
to one another on the map.\footnote{The code for the experiments is available at \url{https://github.com/Project-PRAGMA/Map-of-Rules-IJCAI-2023}.}
While different runs of the algorithm (on the same distance data) may give somewhat different maps, the outputs
tend to be similar, and we believe that the maps give good intuitions regarding the relative distances between the rules.

\subsection{Experimental Setup}
We generated 6000 instances with 100 candidates and 100 voters from
the six following statistical cultures (1000 elections per culture):
1D-Euclidean with $r=0.05$, 2D-Euclidean with $r=0.2$, resampling with
$p=0.1$ and $\phi \in \{0, \frac{1}{999}, \frac{2}{999},\dots,\frac{998}{999},1\}$, disjoint with $p=0.1$, $\phi \in \{0, \frac{1}{999}, \frac{2}{999},\dots,\frac{998}{999},1\}$,
and $g=10$, party-list with $g=10$, and the Pabulib model. 
For all instances, we use a committee size of $k=10$.

\subsection{Experimental Results}
Results for priceability, EJR, PJR, and JR are presented in
Table~\ref{tab:priceability}. Even though only Equal Shares and
\seqphrag \ formally satisfy priceability, for all instances that we
generated, the committees provided by PAV and seq-PAV also satisfied
priceability.
Moreover, \seqphrag{} and seq-PAV satisfy EJR on all instances despite
\seqphrag{} only guaranteeing PJR \cite{mathprog-phragmen} and seq-PAV not even JR \cite{justifiedRepresentation}.
These results indicate that instances witnessing that, e.g., 
seq-PAV fails EJR are rare. In turn, we can conclude that rules 
such as seq-PAV have stronger proportionality properties than their axiomatic analysis reveals.

Let us briefly discuss the existence of cohesive groups in our data, i.e., groups $N\subseteq V$ with $|N|\geq \ell\cdot\frac{n}{k}$ that jointly approves at least $\ell$
  common candidates.
If, for some reason, our data did not contain cohesive groups for $\ell\geq 2$, then the notions of EJR, PJR and JR would coincide; this would render this comparison meaningless.
However, this is not the case. 
For example, for the Disjoint model, 311 of 1000 instances have cohesive groups with $\ell=2$; for the Resampling model this holds for 717 of 1000 instances. Further, for some models, one can see that JR is significantly easier to satisfy than, say, EJR.
This is particularly apparent for CC, which is guaranteed to satisfy JR but fails both PJR and EJR. 
Finally, we note that our results also support the conclusion by \citet{bre-fal-kac-nie2019:experimental_ejr} that---under many statistical cultures---committees satisfying JR often also satisfy PJR and EJR.

For three of our models, that is, disjoint, 1D-Euclidean, and Pabulib,
we present maps of multiwinner voting rules; see
Figure~\ref{fig:maps} (each map is based on $1000$ elections
generated for the respective culture). We start by analyzing the left
map for the disjoint model. On the
``west'' side of the map, we have the CC voting rule; this area can be
interpreted as representing \textit{diversity}: the rules located in
this area select committees that consist of diverse candidates. On the
``east'' side of the map we have AV; this area can be interpreted as
representing \textit{individual excellence}: the rules located here
select committees that consist of individually best candidates, not
taking into account whether the selected candidates are similar to
each other or not. In the central part of the map (highlighted in red), 
we have numerous
rules considered to be proportional, such as the PAV rule or
\seqphrag. 
Note that Geometric rules are close to the ``cluster'' of proportional rules, but
their distance increases for larger $p$-values.
Finally, in the ``south'' area we have Minimax AV, which is
different from all the other rules.
Importantly, this separation of diverse, proportional and individually excellent rules arises naturally in our model.

As for the central map, which is
based on 1D-Euclidean elections, the overall location of the rules is
similar. However, the ``proportional cluster'' is
shifted toward CC. Moreover, SAV is more different from AV than in the
previous map. 
As 1D Euclidean is a simpler model, we generally see fewer details in this map.

Finally, consider the map based on real-life elections from Pabulib.
Here, we see more pronounced differences (as in the disjoint model).
Interestingly, PAV and seq-PAV are almost indistinguishable here.
Generally, it can be observed that the relative position of voting rules does not
vary greatly across different probability models. In particular, the
proportional cluster can be observed in all maps.


\begin{figure*}[t]
\centering
    \includegraphics[width=15.2cm, trim={0.2cm 0.6cm 0.2cm 1.2cm}, clip]{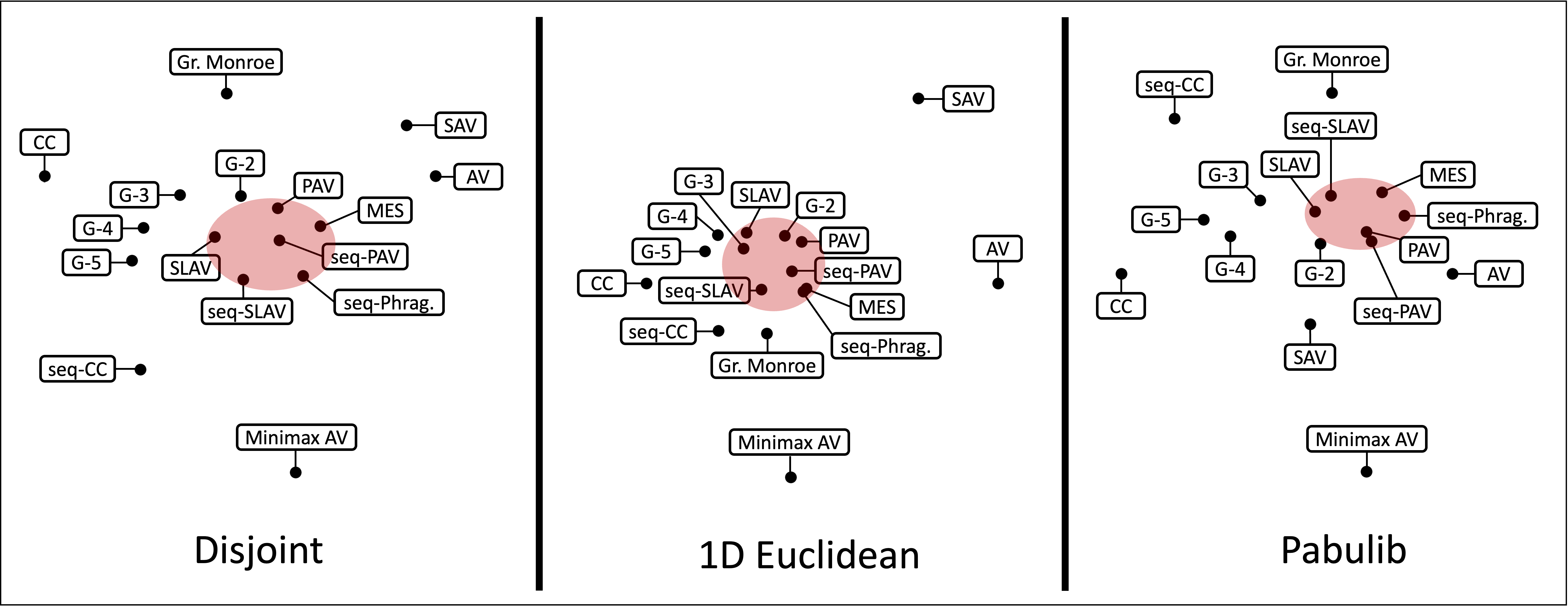}
\caption{Maps of multiwinner rules. Label ``G-$p$'' refers to the $p$-Geometric rule. The red area highlights proportional rules.}
\label{fig:maps}
\end{figure*}

\section{Summary}
We introduced a framework for comparing multiwinner voting rules based on the committees they select.
This framework gives a principled approach to identifying 
similarity of voting rules. When contrasted with axiomatic analysis,
we see that some rules (\seqphrag{} and seq-PAV)
are not only close to strongly proportional rules (Equal Shares and PAV),
but are also indistinguishable from an axiomatic point of view in
our dataset.
This indicates that our  
distance-based 
approach indeed
offers a finer view of multiwinner rules,
and thus complements a precise axiomatic analysis.


\section*{Acknowledgements}
Martin Lackner was supported by the Austrian Science Fund (FWF), research grant P31890.
Krzysztof Sornat was supported by
the SNSF Grant 200021\_200731/1.
We thank Jannik Peters for helpful feedback and
Pasin Manurangsi for a discussion on the hardness of the {\scshape Balanced Biclique} problem. 
This project has received funding from the European Research Council
(ERC) under the European Union’s Horizon 2020 research and innovation
programme (grant agreement No 101002854).
\begin{center}
  \includegraphics[width=3cm]{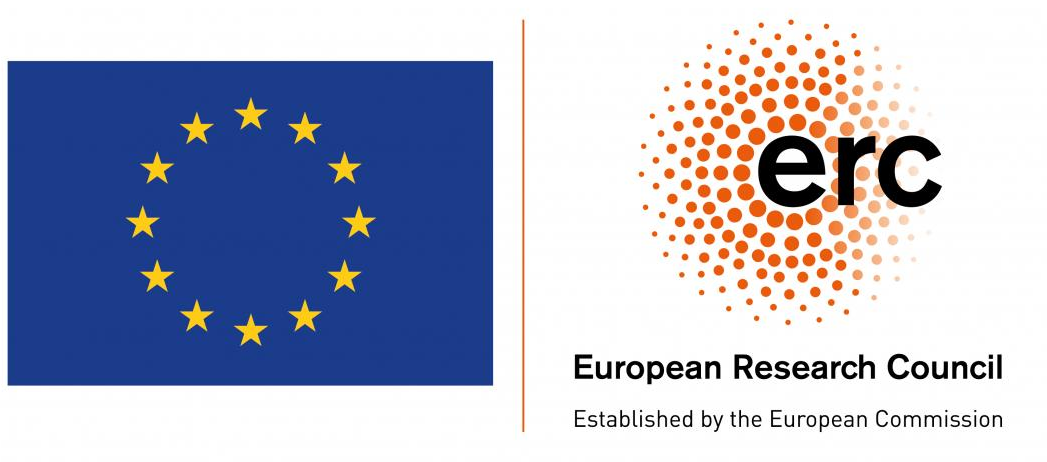}
\end{center}

\bibliographystyle{named}
\bibliography{bib}



\clearpage

\appendix

\begin{center}
  \LARGE Appendix
\end{center}





\section{Proof of Theorem~\ref{thm:fpt-n}}
\label{app:fpt-n}

We will formulate {\scshape Farthest Committees} (FC) as an Integer Linear Program (ILP).

First, we define a set of \emph{types of candidates} as $T = \{t_{V'}: V' \subseteq V \}$, i.e., every subset of voters $V' \subseteq V$ defines a candidate type $t_{V'} \in T$.
We have $|T| = 2^n$.
Naturally, a candidate $c$ has type $t_{A(c)}$.
By $c \in t$ we denote that $c$ has type $t$ and by $|t|$ we denote the number of candidates of type~$t$.
Notice that two candidates of the same type have exactly the same distances to every other candidate.
Formally, for every $a,b \in t$ and $c \in C$ we have $d(a,c) = d(b,c)$.
This holds because of the assumption that $d(x,y)$ is a function of $A(x)$ and $A(y)$ only.

We define two \emph{type-variables} $y_t^L, y_t^R$ for each candidate type $t$.
$y_t^L$ and $y_t^R$ encode how many candidates of type $t$ are included in the first committee and in the second committee respectively.
\begin{align}
  y_t^L, y_t^R \in \{0,1, \dots, \min\{k,|t|\} \} \quad\quad \forall t \in T. \label{eq:y-def}
\end{align}
There are $2^{n+1}$ such type-variables.

We add two constraints meaning that we choose exactly $k$ candidates to each of the committees.
\begin{align}
  \sum_{t \in T} y_t^R = \sum_{t \in T} y_t^L = k.\label{eq:cardinality}
\end{align}
Using thus defined and constrained type-variables, we can define corresponding two committees as an outcome for FC (this is because candidates of the same type are indistinguishable from the objective function's point of view).

Next, we need to encode the matching of the candidates.
For this, we define \emph{edge-variables} $x_{(r,s)}$, for every $r,s \in T$, that encode the number of candidates of type $r$ matched with the same number of candidates of type $s$.
\begin{align}
  x_e \in \{0,1, \dots, k\} \quad\quad \forall e \in T \times T. \label{eq:xe-def}
\end{align}
There are $4^n$ such edge-variables.
The values of the edge-variables define a \emph{b-matching}.

We write the following \emph{b-matching constraints} in order to relate the number of chosen candidates of a certain type and the number of matched candidates of this type.
\begin{align}
  \sum_{r \in T} x_{(t,r)} = y_t^L \quad\quad \forall t \in T,\label{eq:b-matching-l}\\
  \sum_{r \in T} x_{(r,t)} = y_t^R \quad\quad \forall t \in T.\label{eq:b-matching-r}
\end{align}
There are $2^{n+1}$ such constraints.
We say that a b-matching $(\hat{x}_e)_{e \in T \times T}$ is \emph{consistent} with the type-variables values $(\hat{y}_t)_{t \in T}$ if the b-matching constraints are satisfied.

Next we define an objective function of our ILP (although we will add more constraints later).
It is the weight of a b-matching induced by the edge-variables.
\begin{align}
  \mathrm{maximize} \sum_{e \in T \times T} x_e \cdot w_e,\label{eq:ilp-obj}
\end{align}
where $w_{(r,s)}, r,s \in T$ is the distance between candidates, where one is of type $r$ and one is of type $s$.

The ILP defined above is not complete.
In its current form, the output would be the maximum weight matching with exactly $k$ edges, but what we need to achieve is maximum over the choice of the committees of the minimum weight perfect matching between candidates from these committees.
In order to obtain proper output for FC, we keep the objective function exactly as it is defined in Formula~\eqref{eq:ilp-obj}, but we will add more constraints in order to make sure that the chosen b-matching between the candidate types (defined by edge-variables) is a minimum weight b-matching between chosen types of candidates (defined by type-variables).

It is enough to forbid choosing expensive edges (by restricting edge-variables) in case one could change them for cheaper ones while keeping the same values of type-variables.
We will use Lemma~\ref{lem:no-aug-cycles} in order to define proper constraints, but first, we introduce some notation regarding cycles in a bipartite graph.

A cycle $\gamma$ in a complete bipartite graph with parts $(L,R) = (T,T)$ is a sequence of alternating distinct vertices from $L$ and $R$, i.e.:
$$\gamma = (c_1^L, c_1^R, c_2^L, c_2^R, \dots, c_\ell^L, c_\ell^R),$$
where $c_i^L \in L, c_i^R \in R$ and $c_i^X \neq c_j^X$ for every $i \neq j$ and $X \in \{L,R\}$.\footnote{In order to have unique representation of a cycle we assume that, for arbitrary fixed indexing of vertices (types of candidates), the first (resp. the second) vertex in a cycle is an element from $L$ (resp. $R$) with the smallest index.}
A cycle $\gamma$ can be also represented through an edge notation, i.e., $\gamma = M_{LR}^\gamma \cup M_{RL}^\gamma$, where
$M_{LR}^\gamma = \{ \{c_1^L, c_1^R\}, \{c_2^L, c_2^R\}, \dots, \{c_\ell^L, c_\ell^R\} \}$
and
$M_{RL}^\gamma = \{ \{c_1^R, c_2^L\}, \{c_2^R, c_3^L\}, \dots, \{c_{\ell-1}^R, c_\ell^L\}, \{c_\ell^R, c_1^L\} \}$.
Notice that both $M_{LR}^\gamma$ and $M_{RL}^\gamma$ are matchings.
We denote the weight of a matching $M$ as $w(M) = \sum_{e \in M} w(e)$.
Let $\Gamma$ be a set of all cycles in $(L,R)$.
\begin{lemma}\label{lem:Gamma}
  It holds that $|\Gamma| \leq (2^n)^{(2^{n+1} + 1)} \leq 2^{5^n}$.
\end{lemma}
\begin{proof}
  In a bipartite graph a cycle has even length and the longest cycle contains all $2^{n+1}$ vertices.
  There are at most $(2^n)^{(2i)}$ many cycles of length $2i$, where $i \in \{1,2,\dots,2^n\}$.
  We give an upperbound for $|\Gamma|$ as follows.
  $$|\Gamma| \leq \sum_{i=1}^{2^n} (2^n)^{(2i)} \leq \sum_{i=1}^{2^n} (2^n)^{(2^{n+1})} \leq (2^n)^{(2^{n+1}+1)}.$$
  This shows the first inequality in the lemma statement.
  
  The case $n=1$ is polynomial time solvable by taking to one committee the maximum number of approved candidates and taking to the other committee the maximal number of disapproved candidates.
  Therefore $n \geq 2$, which implies that
  $(2^n)^{(2^{n+1} + 1)} \leq 2^{5^n}$.
\end{proof}
In Lemma~\ref{lem:Gamma} we gave an upperbound for $|\Gamma|$ without optimizing constants as the main message of our result is that our algorithm has running time with double-exponential dependence on $n$.
Notice that one can use an algorithm from Theorem~\ref{thm:fpt-n-k} for $n$ upperbounded by a large constant (e.g. $100$) in order to provide a polynomial time algorithm (with a degree above $2^{100}$ in the case of $n \leq 100$).
Then one can assume that $n>100$ and provide better constants in the exponents immediately.
Instead of optimizing constants it is an open question whether a single-exponential algorithm exists.

In Lemma~\ref{lem:no-aug-cycles}, we show that a b-matching has minimum weight (among b-matchings consistent with the given type-variables values) if and only if there is no cycle $\gamma \in \Gamma$ such that the b-matching uses the more expensive matching among $\{ M_{LR}^\gamma, M_{RL}^\gamma\}$.
\begin{lemma}\label{lem:no-aug-cycles}
  For given values of type-variables $(\hat{y}_t)_{t \in T}$ and a b-matching $B = (\hat{x}_e)_{e \in T \times T}$ which is consistent with $(\hat{y}_t)_{t \in T}$, we define a matching
  $M_B = (\min\{1,\hat{x}_e\})_{e \in T \times T}$.
  The following two statements are equivalent:
  \begin{enumerate}
      \item $B$ has minimum weight among b-matchings consistent with $(\hat{y}_t)_{t \in T}$.
      \item For every cycle $\gamma \in \Gamma$ such that $w(M_{LR}^\gamma) \neq w(M_{RL}^\gamma)$ it holds that:
      $$\argmax_{M \in \{ M_{LR}^\gamma, M_{RL}^\gamma \}} \{ w(M) \} \not\subseteq M_B.$$
  \end{enumerate}
\end{lemma}
\begin{proof}
  $(\text{1.}\Rightarrow \text{2.})$ First we show that the first statement implies the second statement.
  By contraposition, let us assume that there exists a cycle $\gamma \in \Gamma$ such that $w(M_{LR}^\gamma) > w(M_{RL}^\gamma)$ (the other case is analogous) and $M_{LR}^\gamma \subseteq M_B$.
  We will show that the first statement of the lemma is false by constructing a new consistent b-matching with smaller weight.
  The idea is to remove the more expensive matching $M_{LR}^\gamma$ from $B$ and add the less expensive matching $M_{RL}^\gamma$.
  
  Let $M_{LR}^\gamma = (\alpha_e)_{e \in T \times T}$ and $M_{RL}^\gamma = (\beta_e)_{e \in T \times T}$, where $\alpha_e, \beta_e \in \{0,1\}$ are b-matching representations of matchings $M_{LR}^\gamma,M_{RL}^\gamma$.
  
  We define $B'$ as follows:
  $$B' = (\hat{x}_e - \alpha_e + \beta_e)_{e \in T \times T} = (u_e)_{e \in T \times T},$$
  and through the following three claims we will show that $B'$ is the desired b-matching.
  \begin{claim}
    $B'$ is a b-matching.
  \end{claim}
  \begin{proof}
    We need to show that
    $u_e \in \{0, 1, \dots, k\}$ for every $e \in T \times T$.
    
    The value $u_e = \hat{x}_e - \alpha_e + \beta_e$ is obviously an integer as $\hat{x}_e \in \{0, 1, \dots, k\}$ and $\alpha_e, \beta_e \in \{0,1\}$.
    
    We have $u_e \geq \hat{x}_e - \alpha_e \geq 0$, where the last inequality follows from the fact that $M_{LR}^\gamma \subseteq M_B$, hence $\hat{x}_e \geq 1$.
    
    We have $u_e \leq \hat{x}_e + \beta_e \leq k+1$.
    We will show that $\hat{x}_e + \beta_e = k+1$ implies a contradiction, hence $u_e \leq k$.
    
    Let us assume that $\hat{x}_e + \beta_e = k+1$.
    It means that $\hat{x}_e = k$.
    B-matching $B$ is consistent with $(\hat{y}_t)_{t \in T}$, so it satisfies b-matching constraints.
    Hence we have:
    \begin{align*}
      \sum_{t \in T} \sum_{r \in T} \hat{x}_{(t,r)} \stackrel{\eqref{eq:b-matching-l}}{=} \sum_{t \in T} y_t^L \stackrel{\eqref{eq:cardinality}}{=} k .
    \end{align*}
    It means that all other edge-variables except $\hat{x}_e$ have value $0$.    
    It implies that $M_B$ contains only one edge, i.e., $e$.
    Moreover, we know that $M_{LR}^\gamma \subseteq M_B$, so $|M_{LR}^\gamma| \leq 1$.
    A cycle $\gamma$ has length at least $2$, so $|M_{LR}^\gamma| \geq 1$.
    Finally, we can say that $M_{LR}^\gamma = \{e\}$.
    But then also $M_{RL}^\gamma = \{e\}$.
    Hence we have $w(M_{LR}^\gamma) = w(M_{RL}^\gamma)$ and this is in contradiction with the assumption that weights of both matchings are different.
  \end{proof}
  \begin{claim}
    $B'$ is consistent with $(\hat{y}_t)_{t \in T}$.
  \end{claim}
  \begin{proof}
    We need to show that both b-matching constraints, i.e. Constraints~\eqref{eq:b-matching-l}--\eqref{eq:b-matching-r}, are satisfied by $B'$.
    We will prove it for Constraint~\eqref{eq:b-matching-l} and the proof for the other constraint is analogous.
    
    First, we notice that for every $t \in L$ the number of edges from $M_{LR}^\gamma$ which are incident to $t$ is the same as the number of edges from $M_{RL}^\gamma$ which are incident to $t$ (this number is equal to $0$ or $1$).  
    Hence we have:
    \begin{align}
      \sum_{r \in T} (\beta_{(t,r)} - \alpha_{(t,r)}) = 0.\label{eq:alpha-beta-0}
    \end{align}
    Then, for every $t \in T$ we have:
    \begin{align*}
      \sum_{r \in T} u_{(t,r)} 
      = \sum_{r \in T} \hat{x}_{(t,r)} + \sum_{r \in T} (\beta_{(t,r)} - \alpha_{(t,r)})
      \stackrel{\eqref{eq:b-matching-l},\eqref{eq:alpha-beta-0}}{=} \hat{y}_t^L.
    \end{align*}
  \end{proof}

  \noindent
  We define the weight of a b-matching $A$ as:
  $$w(A) = \sum_{e \in T \times T} A_e \cdot w_e.$$
  \begin{claim}
    $B'$ has strictly smaller weight than $B$.
  \end{claim}
  \begin{proof}
    The weight of $B'$ is upperbounded as follows:
    \begin{align*}
      w(B') &= \sum_{e \in T \times T} u_e \cdot w_e\\
      &= \sum_{e \in T \times T} \hat{x}_e \cdot w_e - \sum_{e \in T \times T} \alpha_e \cdot w_e + \sum_{e \in T \times T} \beta_e \cdot w_e\\
      &= w(B) - w(M_{LR}^\gamma) + w(M_{RL}^\gamma) < w(B),
    \end{align*}
    where the last inequality follows from the assumption that $w(M_{LR}^\gamma) > w(M_{RL}^\gamma)$.
  \end{proof}
  
  $(\text{2.}\Rightarrow \text{1.})$ Next, we show that the second statement implies the first one.
  By contraposition, let us assume that the first statement is false, i.e., there exists a b-matching $B' = (u_e)_{e \in T \times T}$ consistent with $(\hat{y}_t)_{t \in T}$ such that $w(B') < w(B)$.
  We will construct a cycle which shows that the second statement of the lemma is false.
  
  As both b-matchings $B$ and $B'$ are consistent with $(\hat{y}_t)_{t \in T}$, every vertex of $(L,R)$ has the same number of incident edges in both b-matchings.
  Formally, for every $t \in T$ we have:
  \begin{align*}
    \sum_{r \in T} \hat{x}_{(t,r)} \stackrel{\eqref{eq:b-matching-l}}{=} y_t^L \stackrel{\eqref{eq:b-matching-l}}{=} \sum_{r \in T} u_{(t,r)},\\
    \sum_{r \in T} \hat{x}_{(r,t)} \stackrel{\eqref{eq:b-matching-r}}{=} y_t^R \stackrel{\eqref{eq:b-matching-r}}{=} \sum_{r \in T} u_{(r,t)}.
  \end{align*}
  
  We define two (multi)graphs as follows (for each edge, the value in a vector defines multiplicity of an edge):
  $$(B-B')^+ = ((\hat{x}_e-u_e)^+)_{e \in T \times T},$$
  $$(B'-B)^+ = ((u_e-\hat{x}_e)^+)_{e \in T \times T},$$
  where $x^+ = \max\{x,0\}$.
  The symmetric difference of $B$ and $B'$ is defined as
  $$B \triangle B' = ((\hat{x}_e-u_e)^+ + (u_e-\hat{x}_e)^+)_{e \in T \times T},$$

  $B \triangle B'$ is a (multi)graph on $(T,T)$ with the following property.
  \begin{claim}\label{claim:sym-diff-even}
    Every vertex of $B \triangle B'$ has even degree.
  \end{claim}
  \begin{proof}
    In a graph $B+B' = (\hat{x}_e+u_e)_{e \in T \times T}$, every vertex $t \in L$ has even degree $2y_t^L$ and every vertex $t \in R$ has even degree $2y_t^R$.
    We define $\min\{B,B'\} = (\min\{\hat{x}_e,u_e\})_{e \in T \times T}$.
    Obviously, in a graph $2 \cdot \min\{B,B'\}$ every vertex has even degree.
    Finally, we observe that $B \triangle B'$ is equal to $(B+B')-2\min\{B,B'\}$, so also in $B \triangle B'$ every vertex has even degree.
  \end{proof}
  Claim~\ref{claim:sym-diff-even} implies that $B \triangle B'$ can be covered by cycles from $\Gamma$ (possibly using a given cycle multiple times).
  Let us fix such a cycle cover $\Lambda$.
  
  For a cycle $\gamma \in \Lambda$ we define two matchings:
  $$\gamma_B = \{e \in \gamma: \hat{x}_e > u_e\} = \gamma \cap (B-B')^+,$$
  $$\gamma_{B'} = \{e \in \gamma: \hat{x}_e < u_e\} = \gamma \cap (B'-B)^+.$$
  Notice that $\{\gamma_B, \gamma_{B'}\} = \{M_{LR}^\gamma,M_{RL}^\gamma\}$ and $\gamma_B \neq \gamma_{B'}$.
  Intuitively, $\gamma_B$ is a matching included in b-matching $B$ and $\gamma_{B'}$ is a matching included in b-matching $B'$.
  
  We are ready to show that the second statement of the lemma is false.
  We have the following sequence of (in)equalities:
  \begin{align*}
    0 &< w(B) - w(B') = \sum_{e \in T \times T} (\hat{x}_e - u_e) \cdot w_e\\
    &= \sum_{\substack{e \in T \times T\\ \hat{x}_e > u_e}}
        (\hat{x}_e - u_e) \cdot w_e
      - \sum_{\substack{e \in T \times T\\ \hat{x}_e < u_e}}
        (u_e - \hat{x}_e) \cdot w_e\\
    &= \sum_{e \in (B-B')^+} w_e
      - \sum_{e \in (B'-B)^+} w_e\\
    &= \sum_{\gamma \in \Lambda}
      \bigg(\sum_{e \in \gamma \cap (B-B')^+} w_e
        - \sum_{e \in \gamma \cap (B'-B)^+} w_e\bigg)\\
    &= \sum_{\gamma \in \Lambda}
      \bigg(\sum_{e \in \gamma_B} w_e
        - \sum_{e \in \gamma_{B'}} w_e\bigg)\\
    &= \sum_{\gamma \in \Lambda} \big(w(\gamma_{B}) - w(\gamma_{B'})\big),
  \end{align*}
  where the third equality follows from the fact that every edge $e \in T \times T$ such that $\hat{x}_e > u_e$ appears $\hat{x}_e - u_e$ many times in $(B-B')^+$.
  Analogously in the case $\hat{x}_e < u_e$, an edge $e$ appears $u_e - \hat{x}_e$ many times in $(B'-B)^+$.
  The forth equality follows from the fact that graph $B \triangle B' = (B-B')^+ + (B'-B)^+$ is covered by cycles from $\Lambda$.
  
  Therefore, there exists $\gamma \in \Lambda$ such that:
  $$w(\gamma_{B}) > w(\gamma_{B'}).$$
  
  It means that $B$ uses a more expensive matching among two matchings contained in cycle $\gamma$.
  This shows that the second statement of the lemma is false.
  This finishes proof of the lemma.
\end{proof}

We define a constraint for every possible cycle $\gamma$ in $(T,T)$ which has the following meaning: if two matchings in cycle $\gamma$ have different weights, then a solution cannot use the more expensive one.
In this way, the output matching will be a minimum weight b-matching between chosen types of candidates (defined by type-variables).

In order to define the constraints, we introduce binary \emph{restriction-variables} $z_e$ which inform us if an edge $e$ is used in a b-matching or not:
\begin{align}
  z_e &\in \{0,1\} &\forall e \in T \times T, \label{eq:ze-def}\\
  z_e &\leq x_e &\forall e \in T \times T, \label{eq:ze-leq-xe}\\
  z_e &\geq \frac{1}{k} \cdot x_e &\forall e \in T \times T. \label{eq:ze-geq-1kxe}
\end{align}
There are $4^n$ such restriction-variables.
We introduced $2 \cdot 4^n$ many new constraints.
The behavior of restriction-variables is exactly as intended:
\begin{lemma}\label{lem:z-x}
  For every $e \in T \times T$ it holds:
  \begin{align*}
    z_e = 1 \iff x_e \geq 1,\\
    z_e = 0 \iff x_e = 0.
  \end{align*}
\end{lemma}
\begin{proof}
  If $z_e=1$ then Inequality~\eqref{eq:ze-leq-xe} implies $x_e \geq 1$.
  If $x_e \geq 1$ then Inequality~\eqref{eq:ze-geq-1kxe} implies $z_e \geq \frac{1}{k}$, but $z_e$ is a binary variable, hence $z_e = 1$.
  If $z_e=0$ then Inequality~\eqref{eq:ze-geq-1kxe} implies $x_e = 0$.
  If $x_e = 0$ then Inequality~\eqref{eq:ze-leq-xe} implies $z_e = 0$.
\end{proof}

Now we are ready to define constraints which forbid using more expensive matching if there is a cheaper one.
Let $\Gamma_1 = \{ \gamma \in \Gamma: w(M_{LR}^\gamma) > w(M_{RL}^\gamma) \}$ and $\Gamma_2 = \{ \gamma \in \Gamma: w(M_{LR}^\gamma) < w(M_{RL}^\gamma) \}$, i.e., they contain cycles with different matching weights.
\begin{align}
  \sum_{e \in M_{LR}^\gamma} z_e \leq |M_{LR}^\gamma|-1 \quad\quad \forall \gamma \in \Gamma_1, \label{eq:cycle-restr-lr}\\
  \sum_{e \in M_{RL}^\gamma} z_e \leq |M_{RL}^\gamma|-1 \quad\quad \forall \gamma \in \Gamma_2. \label{eq:cycle-restr-rl}
\end{align}
There are at most $|\Gamma|$ many such constraints, i.e., at most $2^{5^n}$ (due to Lemma~\ref{lem:Gamma}).

Due to the new constraints, at least one variable $z_e, e \in M_{LR}^\gamma$ is equal to $0$ (in the case of $\gamma \in \Gamma_1$; the case $\gamma \in \Gamma_2$ is analogous).
Therefore, due to Lemma~\ref{lem:z-x}, at least one variable $x_e, e \in M_{LR}^\gamma$ is equal to $0$, so at least one edge $e$ from a more expensive matching is not taken to a solution.
Consequently, such more expensive matching is excluded from a set of feasible solutions.
This finishes the description of our ILP.

\citet{lenstra1983integer} showed that an ILP can be solved using $O(p^{2.5p + o(p)} \cdot |I|)$ arithmetic operations, where $|I|$ is the input size and $p$ is the number of integer variables.

In our ILP we have:
$2^{n+1}$ type-variables (Equation~\eqref{eq:y-def});
$4^n$ edge-variables (Equation~\eqref{eq:xe-def});
and $4^n$ restriction-variables (Equation~\eqref{eq:ze-def}).
Hence, in total, there are $O(4^n)$ many variables.

In our ILP we have:
$2$ cardinality constraints (Equation~\eqref{eq:cardinality});
$2^{n+1}$ b-matching constraints (Equations~\eqref{eq:b-matching-l} and~\eqref{eq:b-matching-r});
$2 \cdot 4^n$ constraints from Inequalities~\eqref{eq:ze-leq-xe} and~\eqref{eq:ze-geq-1kxe}; and
$2^{5^n}$ constraints from Inequalities~\eqref{eq:cycle-restr-lr} and~\eqref{eq:cycle-restr-rl}.
Hence, in total, there are $O(2^{5^n})$ many constraints.

Therefore, using the result of \citet{lenstra1983integer} we can find a solution to FC in time upperbounded by:
$$(O(4^n))^{O(4^n)} \cdot O(2^{5^n}) \leq 2^{O(n \cdot 4^n)} \cdot 2^{5^n} \leq 2^{O(5^n)}.$$
This finishes the proof of the theorem.

\end{document}